\documentclass{PoS}

\usepackage{amsthm,amsmath,amssymb}
\usepackage{rotating}

\newcommand{\R}{\mathbb R}

\newcommand{\K}{{\mathbb K}}

\newcommand{\ie}{i.e.,\ }
\newcommand\Mvec{\,\mbox{\bf M}}
\newcommand{\SigmaP}{\texttt{Sigma}}

\newcommand{\HarmonicSumsP}{\texttt{HarmonicSums}}

\theoremstyle{plain}
\newtheorem{theorem}{Theorem}

\newtheorem{prop}[theorem]{Proposition}

\theoremstyle{definition}

\newtheorem{example}[theorem]{Example}

\theoremstyle{remark}

\newcounter{mmacnt}
\def\restartmma{\setcounter{mmacnt}{0}}
\restartmma \catcode`|=\active
\def|#1|{\mathrm{#1}}
\catcode`|=12

\allowdisplaybreaks[4]

\title{Inverse Mellin Transform of Holonomic Sequences}

\ShortTitle{Inverse Mellin Transform}

\author{\speaker{Jakob Ablinger}\thanks{This work has been supported in part by the Austrian Science Fund (FWF) grant SFB F50 (F5009-N15).}\\
        Research Institute for Symbolic Computation (RISC)\\
	Johannes Kepler University Linz, Altenberger Stra\ss e 69, A-4040 Linz, Austria \\
        E-mail: \email{jablinge@risc.jku.at}}        

\abstract{We describe a method to compute the inverse Mellin transform of holonomic sequences, that is based on a method to compute the Mellin transform of holonomic functions. Both methods are implemented in the computer algebra package \HarmonicSumsP.}

\FullConference{Loops and Legs in Quantum Field Theory\\
		24-29 April 2016\\
		Leipzig, Germany}

\begin{document}

\section{Introduction}
\label{sec:1}
In this paper we will present a method to compute the inverse Mellin transform of holonomic sequences and related to it we will revisit a method from~\cite{Ablinger:2014} to compute the Mellin transform of holonomic functions. We emphasize that
these methods are implemented in the computer algebra package \HarmonicSumsP.
Now let $\mathbb K$ be a field of characteristic~0. A function $f=f(x)$ is called \textit{holonomic} (or \textit{D-finite}) if there exist 
polynomials $p_d(x),p_{d-1}(x),\ldots,p_0(x)\in \mathbb K[x]$  (not all $p_i$
being $0$) such that the following holonomic differential equation holds:
\begin{equation}
 p_d(x)f^{(d)}(x)+\cdots+p_1(x)f'(x)+p_0(x)f(x)=0.
\end{equation}
We emphasize that the class of holonomic functions is rather large due to its
closure properties. Namely, if we are given two such differential
equations that contain holonomic functions $f(x)$ and $g(x)$ as solutions, one
can compute holonomic differential equations that contain $f(x)+g(x)$,
$f(x)g(x)$ or $\int_0^x f(y)dy$ as solutions. In other words any composition
of these operations over known holonomic functions $f(x)$ and $g(x)$ is again a
holonomic function $h(x)$. In particular, if for the inner building blocks
$f(x)$ and $g(x)$ the holonomic differential equations are given, also the holonomic
differential equation of $h(x)$ can be computed.\\
Of special importance is the connection to recurrence
relations. A sequence $(f_n)_{n\geq0}$ with $f_n\in\mathbb K$ is called
holonomic (or \textit{P-finite}) if there exist polynomials 
$p_d(n),p_{d-1}(n),\ldots,p_0(n)\in \mathbb K[n]$ (not all $p_i$ being $0$) such
that a holonomic recurrence
\begin{equation}
 p_d(n)f_{n+d}+\cdots+p_1(n)f_{n+1}+p_0(n)f_n=0
\end{equation}
holds for all $n\in\mathbb N$ (from a certain point on).
In the following we utilize the fact that holonomic functions are
precisely the generating functions of holonomic sequences: 
if $f(x)$ is holonomic, then the coefficients 
$f_n$ of the formal power series expansion 
$$f(x) = \sum\limits_{n=0}^{\infty} f_n x^n$$
form a holonomic sequence. Conversely, for a given holonomic sequence
$(f_n)_{n\geq0}$, the function defined by the above sum (\ie its 
generating function) is holonomic (this is true in the sense of formal power series, even if the sum has a zero radius of 
convergence). Note that given a holonomic differential equation for a holonomic function $f(x)$ it is straightforward to 
construct a 
holonomic recurrence for the coefficients of its power series expansion. For a
recent overview of this holonomic machinery and further literature we
refer to~\cite{KauersPaule:2011}.

The paper is organized as follows. In Section \ref{sec:2} we revisit a 
method from~\cite{Ablinger:2014} to compute the Mellin transform of holonomic functions, while in Section \ref{sec:3} we present a method to compute the inverse Mellin transform of holonomic functions.

\section{The Mellin Transform of Holonomic Functions}
\label{sec:2}

\noindent
In the following, we deal with the problem:\\
\textbf{Given} a holonomic function $f(x)$.\\
\textbf{Find,} whenever possible, an expression $F(n)$ given as a linear combination of indefinite
nested sums such that for all $n\in\mathbb N$ (from a certain point on) we have
\begin{equation}\label{Equ:MellRep}
\Mvec[f(x)](n)=F(n).
\end{equation}
In~\cite{Ablinger:2014} three different but similar methods to solve the problem
above were presented. All three methods are implemented in the 
Mathematica package
\texttt{HarmonicSums}~\cite{Ablinger:2010kw,Ablinger:2011te,Ablinger:2013cf,
Ablinger:2013hcp}. All of these methods rely on the holonomic machinery sketched
above. In addition the symbolic
summation package \texttt{Sigma}~\cite{SIG1,SIG2} is used which is based on an
algorithmic difference field theory. 
Here one of the key ideas is to derive a recurrence relation that contains the Mellin
transform as solution and to execute \SigmaP's recurrence solver that finds all
solutions that can be expressed in
terms of indefinite nested sums and
products~\cite{Petkov:92,Abramov:94,Singer:99,Blumlein:2009tj}; these solutions
are called d'Alembertian solutions.
In the following  we revisit one of the methods form~\cite{Ablinger:2014}.

We state the following proposition.

\begin{prop}
 If the Mellin transform of a holonomic function is defined \ie the integral $$\int_0^1x^nf(x)dx$$ exists, then it is 
a holonomic sequence.
\end{prop}

\begin{proof}
Let $f(x)$ be a holonomic function such that the integral $\int_0^1x^nf(x)dx$ exists. Using the properties of the Mellin 
transform we can easily check that
\begin{eqnarray}\label{detore}
\Mvec[x^m f^{(p)}(x)](n)&=&\frac{(-1)^p (n+m)!}{(n+m-p)!}\Mvec[f(x)](n+m-p)+\sum_{i=0}^{p-1}\frac{(-1)^i (n+m)!}{(n+m-i)!}f^{(p-1-i)}(1).
\end{eqnarray}
Finally, we apply the Mellin transform to the holonomic differential equation of
$f(x)$ using the relation above, and we get a 
holonomic recurrence for $\Mvec[f(x)](n).$
\end{proof}

\noindent
Now, a method to compute the Mellin transform is obvious:

\vspace*{2mm}
\noindent
Let $f(x)$ be a holonomic function. In order to compute the Mellin transform 
$\Mvec[f(x)](n),$ we can proceed as follows:
\begin{enumerate}
 \item Compute a holonomic differential equation for $f(x).$
 \item Use the proposition above to compute a holonomic recurrence for $\Mvec[f(x)](n).$
 \item Compute initial values for the recurrence.
 \item \label{RecSol} Solve the recurrence (with \SigmaP) to get a closed form
representation for $\Mvec[f(x)](n).$
\end{enumerate}
\noindent Note that $\SigmaP$ finds all solutions that can be expressed in
terms of indefinite nested sums and products. Hence as long as
such solutions suffice to solve the recurrence in item~\ref{RecSol}, we succeed to compute the Mellin transform 
$\Mvec[f(x)](n).$
\begin{example}We want to compute the Mellin transform of
 $$f(x):=\int_0^x \frac{\sqrt{1-\tau}}{1+\tau} \, d\tau.$$
 We find that
$$(-3+x) f'(x)+2 (-1+x) (1+x) f''(x)=0$$
holds, which leads to the recurrence
\begin{eqnarray*}
6\int_0^1 \frac{\sqrt{1-\tau}}{1+\tau} \, d\tau&=&-2 (n-1) n \Mvec[f(x)](n-2)+3 n \Mvec[f(x)](n-1)+(n+1) (2 n+3) 
\Mvec[f(x)](n).
\end{eqnarray*}
Initial values can be computed easily and solving the recurrence leads to
\begin{eqnarray*}
\Mvec[f(x)](n)&=&\frac{\left(1+(-1)^n\right) \int_0^1 \frac{\sqrt{1-\tau }}{1+\tau } \, d\tau +(-1)^n \left(6+8 \sum\limits_{i=1}^n \frac{(-4)^i}{\binom{2 i}{i}}\right)}{1+n}-\frac{4 (5+4 n) \left(2^n\right)^2}{(1+2 n) (3+2
   n) \binom{2 n}{n}}.
\end{eqnarray*}
\end{example}

Note that this method can be extended to compute regularized Mellin transforms: given a holonomic function $f(x)$ such that 
$$
\int_0^1(x^n-1)f(x)dx
$$
exists, then we can compute
$$
\Mvec[[f(x)]_+](n):=\int_0^1(x^n-1)f(x)dx
$$
using a slight extension of the method above. For example we can compute
$$
\Mvec[[\frac{\log(x)}{1-x}]_+](n)=\int_0^1(x^n-1)\frac{\log(x)}{1-x}dx=\sum_{i=1}^n\frac{1}{i^2}.
$$
\section{The Inverse Mellin Transform of Holonomic Sequences}
\label{sec:3}
\noindent
In the following, we deal with the problem:\\
\textbf{Given} a holonomic sequence $F(n)$.\\
\textbf{Find,} whenever possible, an expression $f(x)$ given as a linear combination of indefinite
iterated integrals such that for all $n\in\mathbb N$ (from a certain point on) we have
\begin{equation*}
\Mvec[f(x)](n)=F(n).
\end{equation*}
As a first step we want to compute a differential equation for~$f(x)$ given a holonomic recurrence for $\Mvec[f(x)](n).$

Analyzing (\ref{detore}) we see that 
\begin{eqnarray}\label{retode}
\Mvec[(-1)^p x^{m+p} f^{(p)}(x)](n)&=&\frac{(n+m+p)!}{(n+m)!}\Mvec[f(x)](n+m)\nonumber\\
&&+\sum_{i=0}^{p-1}\frac{(-1)^{i+p} (n+m+p)!}{(n+m+p-i)!}f^{(p-1-i)}(1).
\end{eqnarray}
Hence we get
\begin{eqnarray}\label{retode1}
n^p\Mvec[f(x)](n+m)&=&\Mvec[(-1)^p x^{m+p} f^{(p)}(x)](n)-a(n)\Mvec[f(x)](n+m)\nonumber\\
&&-\sum_{i=0}^{p-1}\frac{(-1)^{i+p} (n+m+p)!}{(n+m+p-i)!}f^{(p-1-i)}(1),
\end{eqnarray}
where $a(n)\in\K[n]$ with $\deg(a(n))<p.$ We can use this observation to compute the differential equation recursively:
Let
\begin{equation}
 p_d(n)f_{n+d}+\cdots+p_1(n)f_{n+1}+p_0(n)f_n=0
\end{equation}
be the holonomic recurrence for $\Mvec[f(x)](n).$
Let $k:=\max\limits_{0\leq i\leq d}(\deg(p_i(x)))$ and let $c$ be the coefficient of~$n^k$ in the recurrence \ie
$$
c=\sum_{i=0}^d c_i f_{n+i}
$$
for some $c_i\in\K.$ For $0\leq i\leq d$ we replace $c_i n^k f_{n+i}$ by
$$c_i n^k f_{n+i}+c_i(-1)^k x^{k+i} f^{(k)}(x)-c_i\Mvec[(-1)^k x^{k+i} f^{(k)}(x)](n)$$
and apply (\ref{detore}). Considering (\ref{retode1}) we conclude that we reduced the degree of $n.$ 
We apply this strategy until we have removed all appearences of $f_{n+i}.$ At this point we have an equation of the form
$$
 q_l(x)f^{(l)}(x)+\cdots+q_1(x)f'(x)+q_0(x)f(x)+\sum_{j=0}^{k-1}r_j(n)f^{(j)}(1)=0.
$$
where $r_i(n)\in\K[n].$ If all $r_i(n)=0,$ we are done. If not, we differentiate the differential equation. In both cases we end up with a holonomic differential equation for $f(x).$\\
Let us illustrate this strategy using an example.
\begin{example}
Consider the recurrence
\begin{equation}\label{examplerec}
(2+n)f_{n+2}-f_{n+1}-(n+1)f_n=0.
\end{equation}
The maximal degree of the coefficients $f_{n+i}$ with $0\leq i\leq 2$ is $1$ and the coefficient of $n$ of the left hand side of~(\ref{examplerec}) is $f_{n+2}-f_{n}.$ We substitute
\begin{eqnarray*}
 n f_{n+2}&\to&n f_{n+2}- x^{3} f'(x)+\Mvec[x^{3} f'(x)](n)\\
 -n f_{n}&\to&-n f_{n}+ x f'(x)-\Mvec[x f'(x)](n)
\end{eqnarray*}
in (\ref{examplerec}) and apply (\ref{detore}). This yields
\begin{equation}\label{examplerec2}
(-x^3+x)f'-f_{n+2}-f_{n+1}=0.
\end{equation}
since $\Mvec[x^{3} f'(x)](n)=-(n+3)f_{n+2}+f(1)$ and $\Mvec[x f'(x)](n)=-(n+1)f_n+f(1)$.
Next we substitute
\begin{eqnarray*}
 -f_{n+2}&\to&-f_{n+2}- x^{2} f(x)+\Mvec[x^{2} f(x)](n)\\
 -f_{n+1}&\to&-f_{n+1}- x f(x)+\Mvec[x f(x)](n)
\end{eqnarray*}
in (\ref{examplerec2}) and apply (\ref{detore}). Since $\Mvec[x^{2} f(x)](n)=f_{n+2}$ and $\Mvec[x f(x)](n)=f_{n+1},$ this yields the differential equation
\begin{equation}\label{examplerec3}
(-x^3+x)f'(x)-(x^2+x)f(x)=0.
\end{equation}
\end{example}

Our strategy to compute the inverse Mellin transform of holonomic sequences can be summarized as follows:
\begin{enumerate}
 \item Compute a holonomic recurrence for $\Mvec[f(x)](n).$
 \item Use the method above to compute a holonomic differential equation for $f(x).$
 \item Compute a linear independent set of solutions of the holonomic differential\\
       equation (using~\HarmonicSumsP).\label{DiffSol}
 \item Compute initial values for $\Mvec[f(x)](n).$
 \item Combine the initial values and the solutions to get a closed form
representation for~$f(x).$
\end{enumerate}
\noindent Note that $\HarmonicSumsP$ finds all solutions that can be expressed in
terms of iterated integrals over hyperexponential alphabets~\cite{Petkov:92,Abramov:94,Singer:99,Bronstein,Ablinger:2014}; these solutions
are called d'Alembertian solutions. Hence as long as
such solutions suffice to solve the differential equation in item~\ref{DiffSol} we succeed to compute~$f(x).$

\begin{example}We want to compute the inverse Mellin transform of
 $$f_n:=(-1)^n\left(\sum_{i=1}^n\frac{(-1)^i\sum_{j=1}^i\frac{1}{j^2}}{i}-\sum_{i=1}^\infty\frac{(-1)^i\sum_{j=1}^i\frac{1}{j^2}}{i}\right)$$
 We find that
\begin{eqnarray*}
0&=&(n+1) (n+2)^2 f_{n}-(n+2) \left(n^2+7 n+11\right) f_{n+1}\\
&&+\left(-n^3-5 n^2-6 n+1\right) f_{n+2}+(n+3)^3 f_{n+3}
\end{eqnarray*}
which leads to the differential equation
\begin{eqnarray*}
0&=&-(x-1)^2 (x+1) x^3 f^{(3)}(x)-(x-1) (2 x-1) (3 x+1) x^2 f''(x)\\
&&-(x-1) (7 x-1) x^2 f'(x)-(x-1) x^2 f(x)
\end{eqnarray*}
that has the general solution
$$s(x)=\frac{c_1}{x+1}+\frac{c_2}{x+1}\int_0^x\frac{1}{y-1}dy+\frac{c_3}{x+1}\int_0^x\frac{\log(y)}{y-1}dy,$$
for some constants $c_1,c_2,c_3.$ In order to determine these constants we compute 
\begin{eqnarray*}
\int_0^1x^0s(x)dx&=&c_1 \log(2)+c_2 \frac{\log(2)^2-\zeta_2}{2}+c_3 \frac{2 \zeta_3-\log(2) \zeta_2}2,\\
\int_0^1x^1s(x)dx&=&c_1 (1-\log(2))+ c_2 \frac{-\log(2)^2+\zeta_2-2}2+c_3 \frac{\log(2) \zeta_2-2 \zeta_3+2}2,\\
\int_0^1x^2s(x)dx&=& c_1 \frac{2 \log(2)-1}{2}+c_2 \frac{2 \log(2)^2-2 \zeta_2+1}{4}+c_3 \frac{-4 \log(2) \zeta_2+8 \zeta_3-3}{8}.
\end{eqnarray*}
Since
\begin{eqnarray*}
f_0=-\sum_{i=1}^\infty\frac{(-1)^i\sum_{j=1}^i\frac{1}{j^2}}{i}; \ f_1=1+\sum_{i=1}^\infty\frac{(-1)^i\sum_{j=1}^i\frac{1}{j^2}}{i}; \ f_2=-\frac{3}8-\sum_{i=1}^\infty\frac{(-1)^i\sum_{j=1}^i\frac{1}{j^2}}{i}\\
\end{eqnarray*}
we can deduce that $c_0=0,c_1=0\textnormal{ and } c_2=1$ and hence
$$
f_n=\Mvec\left[\frac{1}{x+1}\int_0^x\frac{\log(y)}{y-1}dy\right](n).
$$
\end{example}

Note that the method above only works if the result is of the form
$$c^n \Mvec[f(x)](n)+d$$
for some $c,d\in\R.$
However, in general we will find results of the form
$$\sum_{i=0}^kc_i^n \Mvec[f_i(x)](n)+d$$
for some $c_i,d\in\R.$
Hence in order to deal with more general functions we refine our approach and compute $f_i(x)$ and $c_i$ for $i=1$ to $i=k$ one after another. We illustrate this using the following example.
\begin{example}
 We want to compute the inverse Mellin transform of
 $$f_n:=\sum_{i=1}^n\frac{(-1)^i}{i}\sum_{j=1}^i\frac{1}{2^j j}.$$
 We find that
 \begin{eqnarray*}
0&=&4 (1+n) (2+n) f_n-2 (2+n) (7+2 n) f_{n+1}+\left(2-2 n-n^2\right) f_{n+2}+(3+n)^2 f_{n+3}
\end{eqnarray*}
which leads to the differential equation
\begin{eqnarray*}
0&=&\left(-2 x+3 x^2\right) f(x)+\left(4 x-16 x^2+13 x^3\right) f'(x)+\left(8 x-10 x^2-9 x^3+8 x^4\right) f''(x)\\
&&+\left(4 x^2-4 x^3-x^4+x^5\right) f^{(3)}(x)
\end{eqnarray*}
with the following three linear independent solutions
\begin{eqnarray*}
 &&\frac{1}{2+x},\frac{-\log(1-\frac{x}{2})}{2+x},\frac{1}{4 (2+x)}\biggl(\pi ^2+2 \log ^2(2)+4 \log (1-x) \log (2-x)-2 \log ^2(2-x)\\&&+4 \pi  \log \left(\frac{x}{2}-1\right)+\log \left(\frac{(-2+x)^2}{4 x^2}\right)-4 \text{Li}_2\left(\frac{2}{2-x}\right)+4
   \text{Li}_2(x-1)\biggr).
\end{eqnarray*}
We know that $f_n$ has to appear in at least one of the Mellin transforms of these solutions. Indeed we get
\begin{eqnarray*}
 \Mvec[\frac{-\log(1-\frac{x}{2})}{2+x}](n)&=&(-2)^n\biggl(\sum_{i=1}^n\frac{(-1)^i}{i}\sum_{j=1}^i\frac{1}{2^j j}+\log(2)\biggl(\log(3)-\sum_{i=1}^n\frac{(-1)^i}{i}\\&&+\sum_{i=1}^n\frac{1}{(-2)^ii}\biggr)
 +\frac{\pi^2}{12}-\frac{5\log(2)^2}{2}-\operatorname{Li}_2\left(\frac{1}{4}\right)\biggr).
\end{eqnarray*}
Now we can write
\begin{eqnarray*}
 f_n&=&f_n+\frac{1}{(-2)^n}\Mvec[\frac{-\log(1-\frac{x}{2})}{2+x}](n)-\biggl(\sum_{i=1}^n\frac{(-1)^i}{i}\sum_{j=1}^i\frac{1}{2^j j}+\log(2)\biggl(\log(3)-\sum_{i=1}^n\frac{(-1)^i}{i}\\&&+\sum_{i=1}^n\frac{1}{(-2)^ii}\biggr)
 +\frac{\pi^2}{12}-\frac{5\log(2)^2}{2}-\operatorname{Li}_2\left(\frac{1}{4}\right)\biggr)\\
    &=&\frac{1}{(-2)^n}\Mvec[\frac{-\log(1-\frac{x}{2})}{2+x}](n)-\log(2)\biggl(\log(3)-\sum_{i=1}^n\frac{(-1)^i}{i}+\sum_{i=1}^n\frac{1}{(-2)^ii}\biggr)\\
 &&-\frac{\pi^2}{12}+\frac{5\log(2)^2}{2}+\operatorname{Li}_2\left(\frac{1}{4}\right)
\end{eqnarray*}
Next we compute the inverse Mellin transform of 
$
g_n=\sum\limits_{i=1}^n\frac{1}{(-2)^ii}.
$
We find that
 \begin{eqnarray*}
0&=&-2 (2+n) g_{n+1}+(4+n) g_{n+2}+(3+n) g_{n+3}
\end{eqnarray*}
which leads to the differential equation
\begin{eqnarray*}
0&=&x^2(1-x) g(x)+x^2(2-x-x^2) g'(x)
\end{eqnarray*}
with the solution
$
\frac{1}{2+x}.
$
Note that
\begin{eqnarray*}
  \Mvec[\frac{1}{2+x}](n)=(-2)^n\left(\log(3)-\log(2)+\sum_{i=1}^n\frac{1}{(-2)^ii}\right),
\end{eqnarray*}
hence we can write
\begin{eqnarray*}
 f_n&=&\frac{-1}{(-2)^n}\Mvec[\frac{\log(2-x)}{2+x}](n)+\frac{18\log(2)^2-\pi^2}{12}+\operatorname{Li}_2\left(\frac{1}{4}\right)+\log(2)\sum_{i=1}^n\frac{(-1)^i}{i}.
\end{eqnarray*}
It remains to compute the inverse Mellin transform of 
$
h_n=\sum_{i=1}^n\frac{(-1)^i}{i}.
$
We derive the recurrence
 \begin{eqnarray*}
0&=&-(1+n) h_n+ h_{n+1}+(2+n) h_{n+2}
\end{eqnarray*}
which gives rise to the differential equation
\begin{eqnarray*}
0&=&x(1-x)h(x)+x(1-x^2) h'(x)
\end{eqnarray*}
with the solution
$
\frac{1}{1+x}.
$
Since
\begin{eqnarray*}
  \Mvec[\frac{1}{1+x}](n)=(-1)^n\left(\log(2)+\sum_{i=1}^n\frac{(-1)^i}{i}\right)
\end{eqnarray*}
we finally get
\begin{eqnarray*}
 f_n=\frac{-1}{(-2)^n}\Mvec[\frac{\log(2-x)}{2+x}](n)+(-1)^n\log(2)\Mvec[\frac{1}{1+x}](n)+\frac{6\log(2)^2-\pi^2}{12}+\operatorname{Li}_2\left(\frac{1}{4}\right).
\end{eqnarray*}
\end{example}


\subsection*{Acknowledgements}
I want to thank C. Schneider, C. Raab and J. Bl\"umlein for useful discussions.

\end{document}